\documentclass[letterpaper,twocolumn,10pt]{article}
\usepackage{usenix-2020-09}

\usepackage{booktabs}

\usepackage{paralist}

\usepackage{standalone}

\usepackage[utf8]{inputenc}

\usepackage{amsmath} \usepackage{amsthm} \usepackage{amssymb}
\usepackage{listings}
\usepackage{soul}

\usepackage{tikz}
\usepackage{pgf-umlsd}
\usetikzlibrary{positioning, fit, decorations.pathreplacing, calc,
  shapes.geometric, arrows.meta, decorations.pathmorphing, patterns, shapes.misc, shapes.arrows}
\usepackage{xspace}
\usepackage{tabu}
\usepackage{enumitem}
\usepackage{todonotes}
\usepackage[titletoc]{appendix}

\newcommand{\T}{\ensuremath{\mathcal{A}ggr}\xspace}
\newcommand{\N}{\ensuremath{\mathcal{N\!T}\mathcal{\!P}}\xspace}
\newcommand{\ADM}{\ensuremath{\mathcal{A}dm}\xspace}
\newcommand{\agent}{\ensuremath{\mathcal{A}gt}\xspace}
\newcommand{\dn}{\ensuremath{D}\xspace}
\newcommand{\da}{\ensuremath{D'}\xspace}

\newcommand{\alg}[1]{\mathsf{#1}}
\newcommand{\adv}{\mathcal{A}dv}

\newcommand{\PR}{\mathsf{Pr}}

\newcommand{\riar}{\xrightarrow{\$}}
\newcommand{\ZZ}{\mathbb{Z}}
\newcommand{\WW}{\mathbb{W}}
\newcommand{\GG}{\mathbb{G}}

\newcommand{\DD}{\mathcal{D}}
\newcommand{\HH}{\mathcal{H}}

\newcommand{\RR}{\mathcal{R}}
\newcommand{\OO}{\mathcal{O}}

\newcommand{\advantage}[2]{\alg{adv}^{#1}_{\adv\;\!,#2}}

\newcommand*\Heq{\ensuremath{\overset{\kern2pt H}{=}}}

\newcommand{\eat}[1]{}

\newtheorem{definition}{Definition}
\newtheorem{theorem}{Theorem}
\newtheorem{lemma}{Lemma}

\newenvironment{changemargin}[2]{%
\list{}{\rightmargin#2\leftmargin#1
\parsep=0pt\topsep=0pt\partopsep=0pt}
\item[]}
{\endlist}

\tikzset{cross/.style={cross out, draw=black, minimum size=2*(#1-\pgflinewidth), inner sep=0pt, outer sep=0pt},
cross/.default={1pt}}

\usepackage{biblatex}
\usepackage[available,functional,reproduced]{usenixbadges}

\newcommand{\mypara}[1]{\vspace{.5em}\noindent\textbf{{#1}~}}

\begin{document}
\title{Oblivious Digital Tokens}

\author{
    {\rm Mihael Liskij}\\
    ETH Zurich
    \and
    {\rm Xuhua Ding}\\
    Singapore Management University
    \and
    {\rm Gene Tsudik}\\
    UC Irvine
    \and
    {\rm David Basin}\\
    ETH Zurich
}

\maketitle

\begin{abstract}
\looseness=-1
A computing device typically identifies itself by exhibiting unique measurable behavior or by proving its knowledge of a secret.
In both cases, the identifying device must reveal information to a verifier.
Considerable research has focused on protecting identifying entities (provers) and reducing the amount of leaked data.
However, little has been done to conceal the fact that the verification occurred.

\looseness=-1
We show how this problem naturally arises in the context of {\em digital emblems}, which were recently proposed by the International Committee of the Red Cross to protect digital resources during cyber-conflicts. To address this new and important open problem, we define a new primitive, called an Oblivious Digital Token (ODT) that can be verified obliviously. Verifiers can use this procedure to check whether a device has an ODT without revealing to any other parties (including the device itself) that this check occurred. We demonstrate the feasibility of ODTs and present a concrete construction that
provably meets the ODT security requirements, even if the prover device's software is fully compromised. We also implement a prototype of the proposed construction and evaluate its performance, thereby confirming its practicality.

\end{abstract}

\pagestyle{plain}
\section{Introduction}
\label{sec:intro}
\looseness=-1
We begin with the problem that motivates this research.
Given the rapid growth of cyber armed forces in numerous states and the increasing likelihood of
cyber-warfare, the International Committee of the Red Cross (ICRC) introduced a digital
red cross emblem. According to ICRC, this emblem should \cite{icrc}:
\begin{quote}
``(...) convey a simple message: in times of armed conflict, those who wear them, or
facilities and objects marked with them, must be protected against harm.''
\end{quote}

\looseness=-1
Similarly to prominently displayed red cross emblems physically painted on (or attached to)
protected buildings, facilities, or vehicles in war zones, the digital emblem is intended
to flag digital assets, e.g., computing devices, that are off-limits to cyber-attacks under
international humanitarian law. A possible use-case \cite{icrc} described by the ICRC is as follows:
    \textit{During an armed conflict with State A, State B's malware spreads automatically to affect computers
    managing A's military supplies. State B's reconnaissance indicates that some systems are marked by
    digital emblems and belong to a hospital. State B therefore amends its attack program to avoid harming such systems.}
Among the various technical and operational issues listed by the ICRC, we believe that any digital emblem scheme must satisfy
three important properties for the parties involved in a cyber-conflict (e.g, States A and B) to
respect it\footnote{The fact that some parties may not respect international humanitarian law neither
lessens nor obviates the ICRC's need for digital emblems.}:
\begin{compactitem}
\item\textbf{Verification Obliviousness}~ Emblem verification must not put the aggressor (State B)
at risk of being exposed.
\item\textbf{Binding Integrity}~ Emblems should not be misused by the administrator (State A) ``to falsely
claim protection, e.g., by routing operations through facilities showing the emblem.'' \cite{icrc}
\item\textbf{Security Preservation}~ The security of protected assets should not be weakened by emblems,
i.e., State B should be unable to leverage information returned by emblems
to facilitate or ease its attacks on State A, even if it decides to do this contrary to the law.
\end{compactitem}

\looseness=-1
Since the original call to action on digital emblems by the ICRC, there is increasing interest in this topic,
including within the Internet Engineering Task Force (IETF) \cite{linker-digital-emblem-01}.
Thus far, two realizations of a digital emblem \cite{icrc,adem} have been proposed and more may be on the way.
Both proposals assume that the administrator (State A) is trusted.
One technique embeds information in visible system artifacts \cite{icrc} (e.g., file names, IP addresses, or
domain names) where the names or associated metadata indicate that the entity is protected.
The other technique is the Authenticated Digital Emblem (ADEM) scheme \cite{adem}, which builds a
PKI-like infrastructure for verifying digital emblems inserted into DNS, TLS, and other communication.

\looseness=-1
In practice, we believe that the security of either technique is uncertain. A dishonest administrator can easily
break verification obliviousness by monitoring accesses to a file emblem. It can also compromise binding
integrity by cloning a file emblem to unauthorized devices or by using a device protected by a digital
emblem as a network proxy for its own (not protected) devices' communications. These administrator attacks
are possible due to the inherent tension between binding integrity and (either or both) verification obliviousness
and security preservation. We elaborate on this in the next section.

\looseness=-1
In this paper, we tackle these limitations and propose a novel digital emblem scheme that satisfies three
properties listed above without trusting the system software and the network infrastructure on the administrator's side.
The proposed digital emblem is a dynamically generated proof that a device is protected under international
humanitarian law, in contrast with prior schemes where the emblem is static. To highlight this distinction
and the different trust model, we call our emblems \emph{Oblivious Digital Tokens} (ODTs).

\looseness=-1
At a high level, we devise a way for a device to insert an ODT into every outgoing TLS handshake.
ODT generation takes place inside a pre-installed Trusted Execution Environment (TEE) instantiated
using Intel SGX or Arm TrustZone (currently the most popular commodity TEEs), combining system security
and cryptographic techniques. Disguised as a regular TLS server, an aggressor expects TLS connections
from compromised devices and verifies the ODTs (if any) received over these connections. ODT
verification is made oblivious by ensuring that messages sent by the aggressor are
indistinguishable from a standard TLS message flow.

This work makes the following contributions.
\begin{compactitem}
\item We propose the notion of ODT, which formulates the ICRC's digital emblem concept by considering both sides of a
cyber-conflict as adversaries with respect to ODT security.
\item We present a concrete ODT scheme constructed using two components: a TEE-based system to generate evidence
for binding integrity and a Privacy-Preserving Equality Test (PPET) protocol for evidence verification.
\item We prototype the ODT scheme based on OpenSSL and SGX and evaluate its practicability and performance.
Our results show that ODT generation costs about $144$ms per TLS handshake. For the security evaluation, we use
the Tamarin protocol verifier to prove binding integrity and conduct a statistical and system analysis for
security preservation and verification obliviousness, respectively.
\end{compactitem}
Generally speaking, the ODT scheme can be viewed as a special type of a remote attestation technique.
However, to the best of our knowledge, it is the first-of-its-kind scheme that considers both
the prover and the verifier to be malicious and is thus designed from the stance of a neutral third-party.

\noindent\textsc{Organization.}
Section~\ref{sec:prob} formulates the ODT problem as an extension of digital emblems.
Section~\ref{sec:const} overviews our ODT construction, components of which are described in
Sections~\ref{sec:wg} and~\ref{sec:eq_prot}. The
complete scheme is presented in Section~\ref{sec:integration}. Our {security analysis} is given
in Section~\ref{sec:sec-analysis}, followed by discussion and related work in Section~\ref{sec:discussion}.

\section{Problem Formulation}
\label{sec:prob}
We now review the ICRC digital emblem notion and formulate it as an ODT
under an adversary model that is more realistic and stronger than those of prior schemes \cite{icrc,adem}.

\subsection{ICRC Emblems}
Figure~\ref{fig:red-cross-use-case} depicts the digital emblem setting considered by the ICRC.
There are three stakeholders:

\looseness=-1
\noindent (1) The neutral third party (\N) represents an organization, such as the ICRC or
Doctors Without Borders, that is protected under international humanitarian law. \N issues and installs
emblems on its devices: desktops, servers, laptops, and smartphones. These devices are said to be
\emph{\N-protected} to reflect their special legal status. Analogous to physical emblems, digital emblems
just signal this status and {\it do not prevent} attacks.
Device \dn in Figure~\ref{fig:red-cross-use-case} illustrates one such example.

\noindent (2) Once deployed, i.e., placed in the field, $D$ is managed by an entity on one side of
the cyber-conflict, called the administrator (\ADM). \ADM also manages other non-protected devices without emblems,
e.g., \da in Figure~\ref{fig:red-cross-use-case}.

\noindent (3) The offensive party in the conflict is the aggressor (\T), who mounts cyber-attacks
on \ADM's devices. According to international humanitarian law, \T is expected to
check for the presence of an \N-issued emblem on the targeted or compromised device, before
causing any damage. Hereafter, we use the terms \T and \emph{Verifier} interchangeably.

\begin{figure}[ht]
  \centering
  \begin{tikzpicture}[scale=1, every node/.style={transform shape}]
    \setlength{\abovedisplayskip}{0pt}
    \setlength{\belowdisplayskip}{0pt}

    \node[draw] (operator) {Admin};

    \node[draw, text width=1.8cm, right=1.2cm of operator] (device) {Device \dn};
    \draw[dashed, -{Latex}] ($ (operator.south) + (-0.5cm, 0.0) $) -- ($ (operator.south) + (-0.5cm, -0.5cm) $) -- node[above] {Administrates} ($ (operator.south) + (2cm, -0.5cm) $) --  (device);

    \node[draw, right=2.2cm of device] (threat actor) {Aggressor};
    \draw[{Latex}-{Latex}] ($ (threat actor.west) + (0, 0.1) $) -- node[above] {Connect} ($ (device.east) + (0, 0.1) $);
    \draw[-{Latex}] ($ (device.east) - (0, 0.1) $) -- node[below] (proof) {Proof} ($ (threat actor.west) - (0, 0.1) $);

    \node[draw, above=0.7cm of operator.north, xshift=0.97cm] (NTP) {Neutral third party};
    \draw[-{Latex}] ($ (NTP.east) + (0, 0.1) $) --  ($ (device.north) + (0, 1.1cm) $) -- node[right] (issue) {Issue} (device);

    \node[draw, below=0.7cm of device.south, xshift=-0.1cm] (unprotected device) {Device \da};
    \draw[dashed, -{Latex}] ($ (operator.south) + (2cm, -0.5)$ ) -- (unprotected device);

    \draw[{Latex}-{Latex}] ($ (threat actor.south) + (-0.1, 0) $) -- ($ (unprotected device.east) + (3.23cm, 0.1)$) -- node[above] {Connect} ($ (unprotected device.east) + (0, 0.1) $);
    \draw[-{Latex}] ($ (unprotected device.east) + (0, -0.1) $) -- node[below] {No proof} ($ (unprotected device.east) + (3.45cm, -0.1)$) -- ($ (threat actor.south) + (0.1, 0) $);

    \node[above=of device.north, xshift=2cm] (legend) {};
    \draw[-{Latex[length=1mm,width=1mm]}] ($ (legend) $) -- ++(0.5cm, 0);
    \node[right=0.3cm of legend] {Message};
    \draw[dashed, -{Latex[length=1mm,width=1mm]}] ($ (legend) + (0, 0.4cm) $) -- ++(0.5cm, 0);
    \node[right=0.3cm of legend, yshift=0.4cm] (relation legend) {Relation};
    \node[right=0.3cm of legend, yshift=-0.4cm] (emblem legend) {Emblem};

    \node[cross=4pt, rotate=45, red, right=0.01cm of issue, yshift=-0.1cm] (n1) {};
    \draw (n1.center) circle (5pt);

    \node[cross=4pt, rotate=45, red, right=0.1cm of proof, yshift=-0.1cm] (n2) {};
    \draw (n2.center) circle (5pt);

    \node[cross=4pt, rotate=45, red, left=0.12cm of device.east, yshift=0.12cm] (n3) {};
    \draw (n3.center) circle (5pt);

    \node[cross=4pt, rotate=45, red, left=0.03cm of emblem legend, yshift=0.1cm] (n4) {};
    \draw (n4.center) circle (5pt);

    \node[left=0cm of n4, yshift=4, xshift=6] (help1) {};
    \node[right=0cm of relation legend.north east, yshift=-10, xshift=-12] (help2) {};
    \node[draw, fit=(help1) (help2)] {};

  \end{tikzpicture}
  \caption{Overview of ICRC digital emblem use-case. \dn is \N-protected while \da is not.}
  \label{fig:red-cross-use-case}
\end{figure}

\subsection{Incentives and Models}
\label{sec:adv}
To simplify the description, we only consider two devices, \dn and \da, as in Figure~\ref{fig:red-cross-use-case},
throughout the rest of the paper. The proposed technique is easily generalized to settings with additional devices.

\ADM and \T are hostile to each other and both attempt to exploit the digital emblem scheme to their maximum
advantage, in different ways.
\begin{compactitem}
\item  \ADM has incentives to break binding integrity. It attempts to deceive
\T into believing that \da is \N-protected.
\item \ADM also has incentives to break verification obliviousness. It attempts to detect the existence of emblem
verification activities in \dn and \da, e.g., to determine whether a file read or a network connection is being used
for emblem verification.
\item \T has incentives to verify emblems in order to comply with international humanitarian law. It also aims
to break security preservation by abusing emblem verification to gather information about the target devices'
software and configuration, which is otherwise unavailable.
\end{compactitem}

\looseness=-1
\mypara{Attack Capabilities}
We assume that \T can stealthily hack into both \dn and \da, and inject code using techniques independent
of the digital emblem. \ADM fully controls the network and its infrastructure. This is a more powerful model
than the classic Dolev-Yao adversary \cite{dolevSecurityPublicKey1983} that only controls communications.
\ADM also controls all system software, i.e., operating systems and hypervisors (if any) in
both \dn and \da, and has access to all system resources, e.g., memory and I/O ports.
\ADM monitors all relevant events, such as system calls, network, and disk I/O, as well as
scans all memory for specific patterns.

\mypara{Trust Model}
We assume that \dn is equipped with a hardware-based publicly identifiable TEE, e.g., Intel SGX or Arm TrustZone.
We assume that no adversary against binding integrity can break the security assurances of the TEE,
including data confidentiality and control flow integrity.
However, we do \emph{not} assume that executions inside the TEE reveal no side-channel information to the adversary
that attacks verification obliviousness, since such assurance is outside the TEE's own design scope.
Thus, the TEE is not part of the Trusted Computing Base (TCB) of verification obliviousness, which is realized
using cryptography.

\looseness=-1
We do not consider \T's false claims of emblem absence on \N-protected devices.
Likewise, emblem-unrelated attacks launched from \N-protected devices, e.g., a DDoS attack against \T,
are out of scope of this work.

\noindent\textsc{Caveat.} It is well-known that secret keys maintained within TEEs are subject to
various side-channel attacks \cite{ragab2021crosstalk,sgxpectre2019,oleksenko2018varys}. The compromise of such keys
directly breaks binding integrity since the attacker can then clone the emblem. However, it does not affect
verification obliviousness in our model and is orthogonal to security preservation.

\subsection{Oblivious Digital Token (ODT)}
\label{sec:def}
The ODT definition embodies the ICRC's notion of digital emblems in the setting of
Figure~\ref{fig:red-cross-use-case} with security in the adversarial model described above.
\begin{definition} \looseness=-1
Given any \N-protected device \dn and any unprotected device \da, an ODT is an unforgeable
and verifiable digital emblem created by an emblem generation functionality, such that, if verified successfully,
the ODT proves to \T that \dn is \N-protected and satisfies three properties: (a) it is infeasible to produce an emblem
proving that \da is \N-protected (\textbf{binding integrity}); (b) no entity (including \ADM and \N) can detect the
existence of \T's verification (\textbf{verification obliviousness}); and (c) no information about \dn is leaked
to \T (\textbf{security preservation}).
$\Box$
\end{definition}

\looseness=-1
Formal definitions of the properties (a)--(c) are deferred to Section~\ref{sec:sec-analysis}, which also
contains our security analysis. Note that an invalid ODT does not mean that the device is not \N-protected,\
because this could be the result of  tampering with ODT generation or transmission.
Also note that security preservation requires that an adversary cannot deduce non-public information about
\dn, for example, library versions, running processes or secret keys, by engaging in ODT verification.
Security preservation is not concerned with the anonymity of the device, since protected devices
belong to public institutions.

\subsection{To Invade or Not To Invade}
\label{sec:main-challenge}
\T may wish to verify the ODT of a device just via network communication because, by being non-invasive,
it can easily remain oblivious to \ADM. However, our threat model reflects that \ADM is capable of system
manipulation and we argue that this prevents a non-intrusive \T from  verifying emblems with binding integrity.
To see why this is the case, we describe attacks on ADEM and some hypothetical schemes before delving
deeper into the problem.

\subsubsection{Attacks on ADEM and Its Variants}
Consider a DNS-based ADEM as a warm-up example. \ADM can simply assign \da the IP address ostensibly for \dn.
As a result, \T would mistakenly believe that \da is \N-protected. Similar attacks work against other
types of ADEM emblems.

\looseness=-1
One might try to strengthen ADEM with a hardware-based TEE. For example, \N could install a TEE on \dn to
host a device-specific signing key, with an associated certificate.
For every TLS connection, the TEE  then returns to its peer a dynamically generated emblem,
i.e., an attestation using the current TLS handshake secret and the device state.
This scheme would prevent \ADM from cloning the emblem from \dn to \da, since the TEE's key is bound to
\dn's hardware. However, \ADM can route \da's network traffic through \dn, similar to a NAT setting.
This attack, which is explicitly noted by ICRC, breaks binding integrity since the remote party would
mistakenly treat its communication peer \da as \N-protected.

\looseness=-1
Similarly, isolating the entire TLS layer or even the hardware interface using the TEE
only makes \ADM's routing attack more complex, rather than effectively stopping the attack.

\subsubsection{The Core Issue} Under our adversary model, \ADM can manipulate the system and network
infrastructure so that \emph{no reliable device name\-space} is accessible to \T.
Keys, IP addresses, host names, etc. are all logical identifiers. By themselves, they neither
enable \T to securely pinpoint a physical device nor determine the emblem's protection scope.

\looseness=-1
This issue is somewhat similar to cuckoo attacks \cite{trust,presence_att}, a well-known problem
in identifying the root of trust (RoT) in a remote device.
Current solutions to cuckoo attacks, such as those using distance-bounding \cite{dhar2020proximitee} and
ambient properties \cite{presence_att}, fail in the ODT context for two reasons:
First, they require \T to have physical control over the relevant device, e.g., by being in close physical proximity.
Second, they require the prover's explicit and interactive cooperation with the verifier. Both are
impractical in a cyber-conflict. They also conflict with the verification obliviousness requirement.

\subsubsection{Necessity of Compromise-then-Verify}
\looseness=-1
Recall the use case in ICRC's report \cite{icrc} that suggests a compromise-then-verify paradigm in which \T
invades a target device and injects its \emph{agent}, an executable program for reconnaissance purpose, to check the existence of a file emblem or a process emblem among an ``active process list".
This paradigm is consistent with the so-called ``kill-chain model'' used in the Advanced Persistent Threat (APT)
attacks as reported by Hutchins et al.~\cite{hutchins2011intelligence}, as well as ransomware, and supply chain attacks.
These sophisticated attacks, with several US government organizations (including the DoD) on their victim list,
establish a persistent and stealthy command-and-control channel between a compromised device and a remote server
on the Internet.

\looseness=-1
We thus propose for \T to build the namespace for its targets by leveraging its capability of invading them.
The intuition is that the reconnaissance agent helps \T to identify the target. Namely, \T checks whether the
program's execution and ODT generation take place on the same device.

\mypara{Stealthy Compromise v.s. ODT Properties}
Clearly, \T's stealthy compromise is a prerequisite for its ODT verification to achieve binding integrity.
If the agent is detected, \ADM can run it on \da and enlarge its execution scope to enclose
\dn, thus breaking binding integrity.
The relation between stealthy compromise and verification obliviousness is more intricate and subtle.
It is neither in our scope nor is it \N's responsibility to \emph{ensure} that the program remains undetected.
Nonetheless, it is in our scope to ensure that \emph{ODT verification
does not induce additional risks of exposing \T's reconnaissance agent}.

\section{ODT Scheme: Overview}
\label{sec:const}
The proposed scheme follows the compromise-then-verify paradigm. We assume that \T can stealthily inject its
reconnaissance agent into target devices.
In the following, we first describe how
our scheme expects the agent to behave and then describe the scheme's workflow, followed by the technical approach.

\subsection{Agent Description}
\looseness=-1
The agent either runs as an independent process or resides within a victim process. The ODT scheme only expects the
agent to use its hosting device's default library to establish TLS connections with a remote server, presumably
under \T's control. We chose TLS over other protocols since it is very widely used by applications in laptops and mobile phones, thus providing a better coverage than less popular protocols.
We neither impose requirements on the bytes sent or received in the connections nor do we require the agent
to perform any specific operation or have specific memory contents, such as a key or a secret pattern.

\looseness=-1
Since the proposed scheme only assumes that the agent makes TLS connections in the same way as other
applications and malware do that, we show in Section~\ref{sec:verif-obliv} that the scheme does not
make the agent more susceptible to detection than TLS-using malware.

\subsection{High-Level Workflow}
\looseness=-1
Figure~\ref{fig:system} illustrates a system consisting of \da and \dn managed by \ADM. Processes in both devices
make TLS connections to various servers. \N's TEE (denoted \emph{O-TEE}, for \emph{oblivious TEE}) is installed on \dn.
O-TEE holds a globally unique public key certified by \N and inserts itself into \emph{all} TLS handshakes
initiated by local processes in \dn.  \T's agent stealthily runs on \dn and \da. As O-TEE and \T operate without
coordination, we sketch their operations separately.

\looseness=-1
Treating each incoming TLS server handshake as a challenge,
O-TEE produces the corresponding \emph{co-residence witness}, which serves as evidence that O-TEE is (or is not)
on the same platform as the client process. After the handshake, O-TEE returns to the TLS server the ODT which consists
of the witness and a signature. We stress that O-TEE behaves the same way in every TLS connection and cannot
determine if its peer is \T's TLS server.

\looseness=-1
Upon receiving a TLS client handshake, \T's TLS server (i.e., the second server in Figure~\ref{fig:system}) sends an ODT
verification challenge disguised as the server handshake, such that the resulting flow is indistinguishable from
those sent by regular servers. After the handshake, it extracts the ODT (if any) from received messages.
It then verifies the ODT offline by checking the signature and whether the witness matches the expected state of its agent.

\begin{figure}[!ht]
    \centering
    \includegraphics[width=\columnwidth]{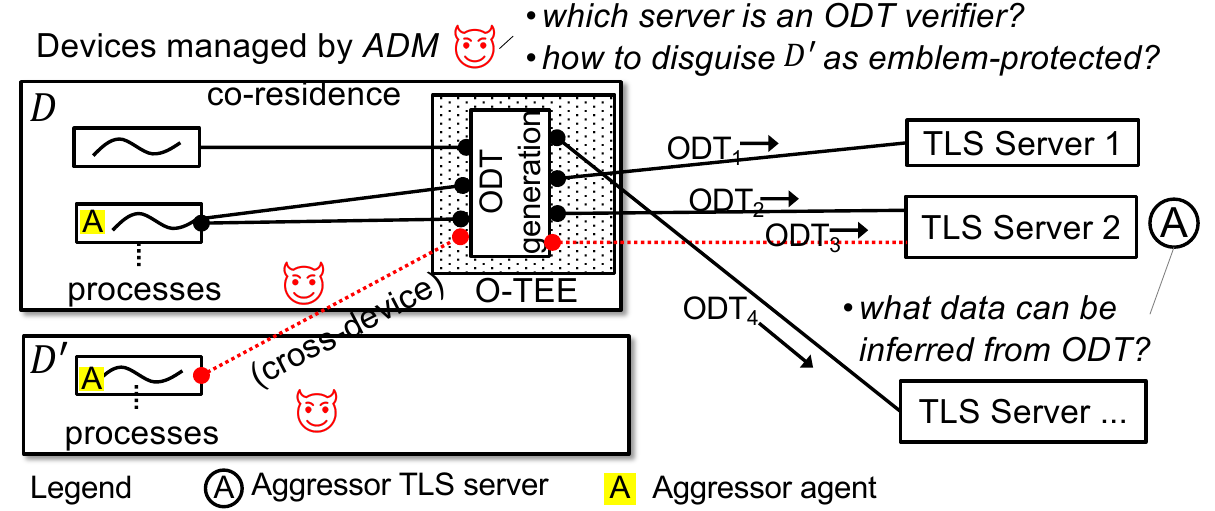}
    \caption{A system view of the ODT scheme. O-TEE mediates four TLS handshake sessions
    and produces distinct ODTs. Red lines indicate \ADM's attack on binding integrity.}
    \label{fig:system}
\end{figure}

\looseness=-1
\subsection{Our Approach}
\label{sec:comp}
Our security rationale is as follows. We design a \emph{Witness Generation System} to attain binding integrity,
and a Privacy-Preserving Equality Test protocol (PPET) for witness verification. The latter prevents an
ODT verifier from inferring information from the witness, thus contributing to security preservation.
Moreover, the witness generation system and the PPET protocol are integrated with TLS messages,
such that neither O-TEE nor \ADM can detect whether an ODT verification is taking place, which ensures
verification obliviousness. We now outline the three components of the ODT scheme whose details are given
in subsequent sections using the notations in Table~\ref{tab:notation}.

\mypara{1. Witness Generation System:} A \emph{witness} is a hash digest
that securely captures both the memory state of the TLS client thread and the thread's co-residence
(or lack thereof) with O-TEE. Essentially, O-TEE runs this component to perform a memory measurement
that attests to binding integrity without trusting the operating system. (See Section~\ref{sec:wg}.)

\looseness=-1
\mypara{2. Privacy-Preserving Equality Test (PPET) Protocol:} The second component is an interactive cryptographic protocol.
It has three message flows between Prover (i.e., O-TEE) holding a secret $w$ and Verifier (i.e., \T) that
guesses Prover's secret as $w'$. At the end of the protocol, Verifier determines whether $w=w'$ is true, without being able to derive any
other information about $w$. The distinguishing feature of this PPET protocol, compared with similar schemes
\cite{lipmaaVerifiableHomomorphicOblivious2003,PET86}, is its \emph{TLS-compatibility}, rather than stronger
security or better performance. This feature means that: (i) protocol messages are compatible with the TLS standard;
and (ii) Prover can successfully run the protocol without differentiating whether TLS handshake messages come from
a normal TLS server or Verifier. This compatibility paves the way for the third component that makes \T
indistinguishable from a normal TLS server, which yields verification obliviousness. (See Section~\ref{sec:eq_prot}.)

\looseness=-1
\mypara{3. TLS Integration} The third component integrates the Prover part of the PPET protocol with the
TLS client handshake protocol run by O-TEE, and integrates the Verifier part of the PPET protocol with the TLS server
handshake protocol run by \T. This integration allows O-TEE to perform PPET
as Prover in \emph{every} TLS handshake, oblivious to whether its peer is an actual Verifier.
Hence \T hides its PPET execution under the guise of a TLS server. (See Section~\ref{sec:integration}.)

\begin{table}
    \centering
    \caption{Notation used throughout the paper.}
    \begin{tabu}{c|p{18em}}
     \tabucline[1pt]{-}
   {\bf Notation}  & {\bf Description} \\ \hline
    $pk_D,sk_D$ & public/private key pair for O-TEE in $D$ \\ \hline
    $\sigma_D$ & signature generated using $sk_D$ \\ \hline
    $\mathbb{B}^{n}$ & set of all $n$-bit binary strings. \\ \hline
    $\Omega$ & predefined set of address ranges, $\Omega\subset \mathbb{B}^{48}$ \\ \hline
    $\mathsf{HS}$ & TLS handshake secret i.e., the ephemeral Diffie-Hellman key \\ \hline
    $\iota^{-1}(), \iota()$ & Elligator~\cite{bernsteinElligatorEllipticcurvePoints2013} encoding and decoding functions. \\
     \tabucline[1pt]{-}
    \end{tabu}
    \label{tab:notation}
\end{table}

\section{Witness Generation}
\label{sec:wg}
Recall that \ADM is the adversary (working against binding integrity) that controls all system software
in \dn and \da. Hence, the main problem for witness generation is to detect
whether the data acquired by O-TEE is \emph{not} staged by \ADM.

\subsection{O-TEE Initialization}
When \dn boots up, its hardware launches O-TEE provisioned by \N.
O-TEE generates an asymmetric key-pair $(pk_D, sk_D)$, the public key component of which ($pk_D$) is certified by the
underlying hardware. \N defines a set of $48$-bit virtual address\footnote{Both x86-64 and ARM64 architectures
use such addresses.} segments to be used as the domain of memory measurement, denoted by
$\Omega\subset \mathbb{B}^{48}$.

\subsection{Memory Measurement Scheme}
\label{sec:mms}
\looseness=-1
The \emph{co-residence witness} is the outcome of the O-TEE-based memory measurement scheme. We define this
scheme below and describe how a valid co-residence witness can only be produced for a process running
on the same device. The measurement is executed by O-TEE after deriving the TLS
handshake secret $k$. Hence, we define the measurement scheme with $k$ as one of the inputs.

\looseness=-1
Suppose that process $P$ in \dn is to be measured. Based on $k$, O-TEE selects one or more virtual memory locations
of $P$ using the function $\mathsf{Select}$, and produces a witness $w$ by running the procedure $\mathsf{Measure}$.
\begin{itemize}[leftmargin=*]
    \item $\mathsf{Select}$ is a function that takes $k$ as input and returns a vector of $m$
    challenge addresses $(c_1,\cdots,c_{m})$ where: $c_i=\mathsf{M}(\mathcal{H}(i||k))$ (for $1<i\leq m$).
    $\mathcal{H}$ is a cryptographic hash function, and $\mathsf{M}$ is a random function
    that uniformly maps a hash digest into $\Omega$.
    \item $\mathsf{Measure}$ is a procedure that reads $m$ memory words $(w_1,\cdots,w_{m})$ from $P$'s
    virtual memory at challenge addresses $(c_1,\cdots,c_{m})$. It also sets a bit-flag $B_{CO}$
    to zero to indicate its co-residence with $P$. It returns a 256-bit $w$ as the \emph{witness}:
    $w=\mathcal{H}(B_{CO}\|w_1\|w_2\cdots\|w_m)$.
    The methods for determining $B_{CO}$ depend on the O-TEE architecture. Two concrete ones are described in Section~\ref{sec:tee}.
 \end{itemize}
The measurement must be \emph{secret}, \emph{atomic} and {\em performed at native speed}. {\em Secret} means that the
memory address to be read is unknown to \ADM before the read operation, which prevents \ADM from gaining any
advantage by preparing the targeted data ahead of time. {\em Atomic} means that, if the memory read is interrupted, the
interruption is faithfully factored into the measurement result. {\em Performed at native speed} means that it is not
slower than the kernel.

\looseness=-1
The rationale for using a secret, atomic and native-speed memory measurement scheme is as follows.
Consider a process $\hat{P}$ running on \da. \ADM wants to deceive O-TEE about its co-residence with
$\hat{P}$. Under our adversary model, \ADM cannot tamper with O-TEE's execution. It also does not know any
challenge addresses before the measurement starts. Hence, it either (i) replicates $\hat{P}$'s entire virtual memory
to \dn, or (ii) tampers with O-TEE's measurement to learn the challenge addresses before making a targeted copy.

\looseness=-1
The former (i) is a brute-force attack that incurs a prohibitively high cost since the size of an application's
virtual memory ranges from a few to hundreds of megabytes. It requires copying the memory words
as well as building the same address mapping, which involves page table modifications.
The cost is even higher considering that \ADM does not know which process on \da is \T's agent.
It must clone every process from \da to \dn in order to pass ODT verification
with a 100\% probability. This is equivalent to running them directly on \dn.

\looseness=-1
For attack (ii) to succeed, \ADM must: (1) win the race with O-TEE (i.e., complete the forgery before
O-TEE finishes the read), and (2) hide any trace of tampering. O-TEE always starts the measurement before
\ADM can attack, and O-TEE reads the memory as fast as \ADM. On one hand, without interrupting O-TEE,
\ADM cannot win because it must first observe and then copy. On the other hand, an interrupt can be detected
by O-TEE. Moreover, since all memory locations are independently chosen $c_i=\mathsf{M}(\mathcal{H}(i||k))$, similar to shuffled measurements \cite{GRT18} used in device attestation,
exposure of the locations measured thus far reveals nothing about subsequent ones.

\looseness=-1
In summary, if the O-TEE-based memory measurement scheme is secret, atomic and performed at native speed,
it yields a valid co-residence witness with respect to the measured process $P$ and O-TEE.
The concrete schemes described in Section~\ref{sec:tee} satisfies the three requirements above.
O-TEE uses the secret key $k$ to determine the memory location(s) to measure.
It reads the target memory at the native speed and in an atomic way.
The hardware updates the O-TEE state if an interrupt or exception occurs when executing the memory read instruction.

\subsection{Concrete Schemes}
\label{sec:tee}
\looseness=-1
On laptops, desktops, and servers, we propose to instantiate O-TEE using the well-known Intel SGX\footnote{While
SGX is deprecated on future ``client platform" processors, it continues to be available on Xeon processors for
servers and cloud platforms.}. On mobile and IoT devices, we instantiate O-TEE using Arm TrustZone on Cortex-M or
Cortex-A processors. Confidential computing technologies, such as Intel TDX \cite{tdx}, Arm CCA, and AMD SEV,
are not ideal for O-TEEs because it is difficult (if not impossible) for software running in those TEEs to
atomically access virtual memory of external processes.

\subsubsection{SGX-based O-TEE}
\label{sec:enclave}
Intel SGX isolates user space data and code demarcated by a virtual address range with an \emph{enclave}
comprising a set of Enclave Page Cache (EPC) pages. SGX provides a confidential computing environment:
no software outside an enclave (e.g., kernel or BIOS) can access an enclave's EPC pages.
Also, an SGX enclave is publicly verifiable.

\looseness=-1
O-TEE can take the form of an enclave mapped to all processes' virtual address spaces as part of the system library.
Intel's Attestation Service (IAS) provides a hardware-based remote attestation facility that allows a remote
verifier to authenticate an enclave's genuineness based on the
hash digest of memory pages loaded during enclave creation.
O-TEE encloses the hash of its public key $pk_D$ into its local attestation to the Quoting
Enclave, an enclave provisioned by Intel to facilitate remote attestation. The latter generates
and signs the enclave quote for O-TEE using the platform's attestation key. The signed quote
is released for public verification.
Hence, \T can determine that there \emph{exists} a legitimate O-TEE instance that owns $pk_D$.

\mypara{Atomic Memory Access} \looseness=-1
SGX does not provide a secure clock for the enclave code to measure time.
When executing enclave code, the CPU accesses non-EPC pages the same way as when running normal code,
i.e., via the MMU's address translation and the memory controller. A memory load
instruction is either executed unobstructed or it encounters an exception. No software can
intervene in its operation without triggering an exception.
Any exception during enclave execution causes an immediate enclave exit.
Before the CPU is trapped to the kernel, an \texttt{EXITINFO} object located in the enclave State
Save Area (SSA) page is updated by hardware to reflect the exception cause. Since the SSA page is
inside the enclave, no system software can read from, or write to, it. Our design uses
this SGX feature to detect exceptions during memory measurement.

\mypara{Measurement and Witness Computation} O-TEE enclave is pre-configured to enable the
hardware to log page faults and general protection exceptions. (Other types of exceptions are logged
by default.) It clears \texttt{EXITINFO} in the SSA page just before the measurements. It then sets
$B_{CO}$ to 0 after the measurement if \texttt{EXITINFO} indicates an exception and sets $B_{CO}$ to 1 otherwise.

\subsubsection{TrustZone-based O-TEE}
\label{sec:tz} \looseness=-1
Arm TrustZone partitions the platform into: {\it Normal World} that hosts the Rich OS plus regular
applications and {\it Secure World} that hosts the Trusted OS and secure applications.
O-TEE is realized as a kernel module of the Trusted OS. Unlike Intel SGX, there is no centralized
trust hierarchy established on a processor and software in Secure World. Nonetheless, the Trusted OS can
carry a credential certified by \N or the manufacturer. The Trusted OS on \dn can further certify
$pk_D$ generated by O-TEE. Thus, the certificate of $pk_D$ is publicly verifiable.

\mypara{Atomic Memory Access}
As part of the Trusted OS, O-TEE runs at Exception Level 1 (EL1) of Secure World. Its code and data are
located in the kernel virtual address range and are thus translated via the \texttt{TTBR1\_EL1} register.
When execution of a user process $P$ is trapped to Secure World for requesting TLS connections,
O-TEE clones the \texttt{TTBR0\_EL1} register serving $P$. This way, O-TEE memory loading instructions
can directly read $P$'s virtual memory.

\looseness=-1
All exceptions raised in Secure World are delivered to the Trusted OS by the hardware. To detect an
exception during measurement, O-TEE hooks the exception handler of the Trusted OS before the measurement starts.
It sets $B_{CO}$ to 0 if its handler is called; and to 1 otherwise.

\section{Privacy-Preserving Equality Test (PPET)}
\label{sec:eq_prot}
\looseness=-1
Recall that PPET is a building block for the ODT scheme. It
allows \T to check whether the witness $w$ generated by O-TEE matches the value $w'$ expected by \T.
We now describe cryptographic aspects of the protocol, showing how the privacy of $w$
is protected against \T when $w\neq w'$ and how it is structurally compatible with TLS.

\looseness=-1
While prior private equality test (PET) \cite{lipmaaVerifiableHomomorphicOblivious2003,PET86}
and private set intersection (PSI) \cite{PSI,PSI_tsudik,PSI99} protocols provide the privacy
properties we need, they are \emph{not} TLS-compatible. For instance, some protocols require the
verifier to send two messages \cite{lipmaaVerifiableHomomorphicOblivious2003} or a zero-knowledge
proof showing that its message is well-formed \cite{PSI_tsudik}. Since large messages and complex
cryptographic data structures do not fit into TLS messages, \T cannot send them covertly.
This leads us to construct a TLS-compatible PPET protocol.

\subsection{The Protocol}
Let $p$ be a large prime number and $\GG$ be a group of order $p$ where the DDH assumption \cite{ddh} holds.
Let $g$ be a generator in $\GG$. The PPET protocol $\mathit{PPET}$
consists of four steps run by Prover and Verifier as in Figure~\ref{fig:wv}.

\begin{figure}[ht]
    \centering
    \begin{tabular}{|c|}
    \hline
    \begin{minipage}{.95\columnwidth}
        \small
      \strut{\tt \ul{PPET protocol between Prover holding $w$ and\\Verifier holding $w'$}}
       \begin{enumerate}[leftmargin=*]
        \item (by Prover) Send a random group element $u \in_R \GG$.
        \item (by Verifier) Do the following:
        \begin{itemize}
            \item Pick a random number $s \in_r \ZZ_p$ and compute $v=g^su^{w'}$ as the commitment to its $w'$.
            \item Send $v$.
        \end{itemize}
        \item (by Prover) If $v$ is not in $\GG$, send two random group elements $y,z\in_R \GG$; otherwise do the following:
        \begin{itemize}
            \item Pick a random number $t \in_r \ZZ_p$ and compute $y=g^t$ and $z=v^tu^{-wt}$.
            \item Send $(y,z)$.
        \end{itemize}
        \item (by Verifier) Assert $w=w'$ iff $z=y^s$.
      \end{enumerate}
    \end{minipage} \\
    \hline
  \end{tabular}
  \caption{PPET Protocol $\mathit{PPET}$.}
  \label{fig:wv}
\end{figure}

\looseness=-1
\noindent If and only if $w'=w$, $z=g^{st}u^{(w'-w)t}= g^{st} = y^s\,.$
Hence, if both parties execute the protocol faithfully, the protocol returns the correct outcome with
respect to $w$ and $w'$.
The rationale behind $\mathit{PPET}$ is that Verifier commits to its witness $w'$ as part of sending $v$.
Prover's response is computed from $v$ such that a mismatching $w'$ is not canceled out of the $u^t$ term,
which masks information about $w$. As a result, Verifier must run the protocol with Prover to verify its
guessed value for $w$. If Verifier intends to learn $w$, it must perform \emph{online} guessing attacks.
By keeping $w$ secret, we minimize information leakage about the memory locations used to derive it.

\subsection{Privacy Preservation}
If, after running $\mathit{PPET}$, Verifier learns that $w'=w$, information is leaked.
However, if $w'\neq w$, Verifier learns nothing about $w$ except that it differs from $w'$.
Thus, a malicious Verifier can eliminate at most one value after each protocol
run and cannot perform an offline guessing attack. We denote the resulting security property
as \emph{privacy preservation}.

\begin{definition}[Privacy Preservation]
\label{def:witn-indist}
  Let $\GG$ be a group of order $p$ and $g$ a group generator.
  Let $\WW$ be the domain of possible witnesses.
  For all witness verification protocols $\mathit{PPET}$ and all adversaries $\adv$,
  we define the advantage function:
  \begin{equation*}
    \begin{split}
      \advantage{\mathsf{PP}}{\mathit{PPET}} = \; & \Big| \frac{1}{2} - \PR \big[ \GG \riar u; \ZZ_p \riar t; \WW \riar w; \\
       & \adv(u) \riar v \in \GG; z_0 = v^t u^{-wt}; \GG \riar z_1; \\
    & \{0, 1\} \riar b; \adv(g^t, z_b) \riar b': b' == b \big] \Big|.
    \end{split}
  \end{equation*}

A $\mathit{PPET}$ satisfies privacy preservation if, for all efficient algorithms $\adv$,
$\advantage{\mathsf{PP}}{\mathit{PPET}}$ is negligible when $\adv$ incorrectly guesses the witness.
\end{definition}

\begin{theorem}
\label{thm:privacy-preservation}
   $\mathit{PPET}$ satisfies privacy preservation. \footnote{For the proof, see the appendix.}
\end{theorem}

\subsection{Structure Compatibility with TLS}
The message flows of the $\mathit{PPET}$ protocol and a TLS handshake have a similar composition structure. (The distribution indistinguishability between corresponding data objects is addressed in the next section.)
Figure~\ref{fig:standard-tls-handshake} illustrates
the client-server interaction in a TLS v1.3 handshake~\cite{TLS13}. The first two messages perform an ephemeral
Diffie-Hellman key exchange. After receiving $\mathsf{ClientHello}$, the server computes the handshake
secret $\mathsf{HS}=X^{y_s}$ before sending $\mathsf{ServerHello}$.
\begin{figure}[ht]
  \centering
    \scalebox{0.85}{
  \begin{tikzpicture}[every node/.style={transform shape}]
    \pgfmathsetmacro{\myrow}{-0.41};
    \pgfmathsetmacro{\myspace}{0.1};

    \node[draw, inner sep=2px] (enclave) {TLS client};
    \draw (enclave) -- ++(0cm, -1.24cm) -- +(0.5cm, 0) -- +(-0.5cm, 0);

    \node[draw, right=5.5cm of enclave, inner sep=2px] (server) {TLS Server};
    \draw (server) -- ++(0cm, -1.24cm) -- +(0.5cm, 0) -- +(-0.5cm, 0);

    \begin{scope}
        \draw[->] ($ (enclave) + (0, \myrow * 1 + \myspace) $) -- node[above, above, inner sep=0] (first message) {$\mathsf{ClientHello}(n_0,X=g_{\text{DH}}^{x_s})$} ($ (server) + (0, \myrow * 1 + \myspace) $);

        \draw[->] ($ (server) + (0, \myrow * 2 + \myspace) $) -- node[above, inner sep=0] (second message) {$\mathsf{ServerHello}(n_1,Y=g_{\text{DH}}^{y_s}), \mathsf{Cert}, \mathsf{Finished}$} ($ (enclave) + (0, \myrow * 2 + \myspace) $);

        \draw[->] ($ (enclave) + (0, \myrow * 3 + \myspace) $) -- node[above, inner sep=0] {$\mathsf{Finished},  \mathsf{Application Data}$} ($ (server) + (0, \myrow * 3 + \myspace) $);

    \end{scope}

  \end{tikzpicture}
  }
    \caption{Message flows in the TLS v1.3 handshake.}
    \label{fig:standard-tls-handshake}
\end{figure}

\mypara{Verifier Flow vs. TLS Server Flow}
\looseness=-1
$\mathit{PPET}$ requires Verifier to send one random group element $v$, while in TLS a server sends a nonce.
The congruence between the two flows allows us to make them indistinguishable using additional encoding
techniques (see Section~\ref{sec:integration}), which are needed for oblivious verification. Note that
it is insecure to use TLS application data to send $v$ since it is \emph{not} a universal behavior
for all TLS servers.

\looseness=-1
\mypara{Prover Flows vs. TLS Client Flows} The first flow in $\mathit{PPET}$ from Prover is congruent to
its counterpart in TLS $\mathsf{ClientHello}$. Both include random numbers. However, TLS client's \texttt{FINISHED}
message does not involve any random numbers. Hence, when $\mathit{PPET}$ is integrated with TLS, $y$ and $z$ are sent
in a TLS heartbeat message \cite{TLS13}.


\section{Complete ODT Scheme Over TLS}
\label{sec:integration}
\looseness=-1
We now present the ODT scheme that integrates: (1) witness generation and Prover steps of $\mathit{PPET}$ into
TLS client functionality inside O-TEE on \N-protected devices, and (2) Verifier steps of $\mathit{PPET}$ into \T
acting as a TLS server. This integration ensures that the \T's message is indistinguishable from a $\mathsf{ServerHello}$
message.  Hence, O-TEE runs consistently with standard TLS servers and \T's server.

\subsection{TLS for ODT Communication}
\label{sec:covert}
While the $\mathit{PPET}$ and TLS flows are structurally similar, we must fine-tune
the data representation in order to achieve indistinguishability.

\subsubsection{TLS Placeholders}
For ODT messages from O-TEE, we use the field $\mathsf{ClientHello.random}$ (i.e., client nonce $n_0$ in
Figure~\ref{fig:standard-tls-handshake}) to send $u$ and the TLS heartbeat request for sending
$y$ and $z$. We chose to send the message $v$ from \T in the field $\mathsf{ServerHello.random}$,
i.e., server nonce $n_1$ in Figure~\ref{fig:standard-tls-handshake}. This chosen field is mandatory
in the TLS handshake phase. Hence, its appearance in the flow does not violate obliviousness.

\looseness=-1
Both nonces are 256 bits long and, more importantly, both are randomly sampled values, so no
subsequent TLS computations restrict their domain. By default, a TLS heartbeat request carries a payload
of up to $2,048$ bytes and requires a heartbeat response echoing it back to the sender.
The table below summarizes the TLS objects
involved in the scheme.

\begin{table}[ht]
    \centering
    \begin{tabu}{llll}
     \tabucline[1pt]{-}
   {\bf TLS obj}  & {\bf Type in TLS} & {\bf ODT obj} & {\bf Sent by} \\ \hline
   nonce $n_0$ & random 256 bits & $u$  & O-TEE  \\
   heartbeat & 2048 bytes & $y,z$ & O-TEE \\
   nonce $n_1$ & random 256 bits & $v$ & Server /\T \\
           \tabucline[1pt]{-}
    \end{tabu}
\end{table}

\looseness=-1
\subsubsection{\texorpdfstring{$\mathit{PPET}$}{PPET} Data in TLS Objects}
To load $\mathit{PPET}$ data into selected TLS objects, we must address both the size issue
(to fit the limited binary size) and the binary representation issue (to match the distribution of the TLS objects).
Since both nonces are 256 bits, we instantiate group $\mathbb{G}$ in PPET as an elliptic curve
whose group elements can be represented by 256-bit strings. As a result, both $u$ and $v$ fit in their respective nonces. Since the heartbeat request can accommodate up to $2,048$ bytes, it can easily hold $y$ and $z$.

\mypara{Group Elements vs. Random Numbers}
\looseness=-1
As described in Section~\ref{sec:eq_prot}, \T's commitment $v$ is uniformly distributed in $\GG$. Since $\GG$ is
instantiated as a subgroup of points on an elliptic curve, $v$ must be converted into a 256-bit string.
Hence the set of the binary strings representing all group elements is only a subset of the domain $\mathbb{B}^{256}$,
which is used for nonce generation. Since the server nonce is not always convertible to an element in $\mathbb{G}$,
directly using $v$ as a nonce in $\mathsf{ServerHello}$ indicates that there is a non-negligible chance that the server is \T.

\looseness=-1
To address this issue, we use Elligator~\cite{bernsteinElligatorEllipticcurvePoints2013} -- a censorship-circumvention
technique that transforms elliptic curve points into random-looking strings and vice versa -- so that an adversary cannot
differentiate between an elliptic curve point and a random string with non-negligible probability.
We use Curve25519~\cite{bernsteinCurve25519NewDiffieHellman2006} for $\GG$ since it is supported by Elligator.
To represent $v$ with a 256-bit nonce, we first apply the Elligator decoding function which returns a 254-bit binary string,
and then prepend the string with two random bits to get the nonce. To parse a 256-bit nonce to $v$, we discard
the two most significant bits and apply the Elligator encoding function to the remaining bits.

\textsc{Caveat.} Note that $\GG$ is used by both O-TEE and \T in the $\mathit{PPET}$ protocol. It is independent
of the group used in the Diffie-Hellman handshake in TLS. Moreover, we do not require O-TEE's message $u$
in the $\mathit{PPET}$ protocol to be encoded using Elligator, since O-TEE always sends out a random group element
regardless of whether its peer is \T or not. There is therefore no risk to verification obliviousness.

\subsection{A Complete Account of the ODT Scheme}

\looseness=-1
Figure~\ref{fig:public-odt-message-diagram} shows the computational steps of the ODT scheme. We first describe how
O-TEE runs as a TLS client and emits an ODT in every TLS handshake. Afterwards, we describe how \T runs as
a TLS server and verifies an ODT.

\begin{figure}[ht]
  \centering
  \scalebox{0.65}{
  \begin{tikzpicture}[every node/.style={transform shape}]
    \pgfmathsetmacro{\myrow}{-0.8};

    \node[draw] (enclave) {O-TEE as TLS client};
    \draw (enclave) -- ++(0cm, -6cm) -- +(0.5cm, 0) -- +(-0.5cm, 0);

    \node[draw, right=5.5cm of enclave] (server) {TLS Server (Aggressor)};
    \draw (server) -- ++(0cm, -6cm) -- +(0.5cm, 0) -- +(-0.5cm, 0);

    \begin{scope}
        \draw[->] ($ (enclave) + (0, \myrow * 1) $) -- node[above] (first message) {$\mathsf{ClientHello}\{n_0=u,X=g_{\text{DH}}^{x_s}\}$} ($ (server) + (0, \myrow * 1) $);

        \node[left=2.5cm of first message] {$\GG \riar u$};

        \draw[->] ($ (server) + (0, \myrow * 2) $) -- node[above] (second message) {$\mathsf{ServerHello}\{n_1=\iota(v),Y=g_{\text{DH}}^{y_s}\}, \mathsf{Cert}, \mathsf{Finished}$} ($ (enclave) + (0, \myrow * 2) $);

        \node[right=1cm of second message] {\shortstack[l]
          {
            $\ZZ_p \riar s$; \\
            $v = g^su^{\bar{w}}$
          } };

        \draw[->] ($ (enclave) + (0, \myrow * 2.6) $) -- ++(0.3, 0) -- node[right] {\shortstack[l]
          {
            Derive $\mathsf{HS}=Y^{x_s}$; \\
            $k=\mathsf{HKDF}(\mathsf{HS})$; \\
            $c=\mathsf{Select}( k)$; \\
            $w \leftarrow \mathsf{Measure}(c)$; \\
            $v=\iota^{-1}(n_1)$; \\
            $\ZZ_p \riar t; y = g^t; z = v^tu^{-wt};\sigma=\mathsf{Sign}(y,z,H(k))$
            } } ++(0, -2.0) -- ++ (-0.3, 0);

          \draw[->] ($ (enclave) + (0, \myrow * 6.5) $) -- node[above] {$\mathsf{Finished}, \mathsf{HeartbeatRequest}\{y, z,\sigma\}$} ($ (server) + (0, \myrow * 6.5) $);

        \draw[->] ($ (server) + (0, \myrow * 6.7) $) -- ++(-0.3, 0) -- node[left] {Continue data exchange} ++ (0, -0.4) -- ++(0.3, 0);

    \end{scope}

  \end{tikzpicture}
  }
    \caption{Process $P$'s TLS handshake through O-TEE.}
    \label{fig:public-odt-message-diagram}
\end{figure}

\subsubsection{O-TEE as TLS client}
\label{sec:client}
Triggered by process $P$'s request to establish a TLS connection, O-TEE executes the following steps.

\noindent\textsc{Step 1.}
It initiates a TLS handshake session by sending $\mathsf{ClientHello}$. It sets nonce
$n_0$ as the binary representation of a random group element $u\in_R\GG$. This is the first message in $\mathit{PPET}$.

\vspace{.5em}
\noindent\textsc{Step 2.}
After receiving $\mathsf{ServerHello}$ with nonce $n_1$, it computes the handshake secret
$\mathsf{HS}$ following the TLS specification and runs the witness generation procedure in Section~\ref{sec:wg} against
$P$'s virtual memory. It also treats $n_1$ as the Verifier message in $\mathit{PPET}$ and decodes $n_1$ into $v\in \GG$. It sets
\[
k\!=\!\mathsf{HKDF}(\mathsf{HS});
c\!=\!\mathsf{Select}(k);
w\!=\!\mathsf{Measure}(c);
v\!=\!\iota^{-1}(n_1\!);
\]
where $\mathsf{HKDF}$ is the key derivation function used in TLS and $\iota^{-1}$ is the Elligator encoding function used to
extract an elliptic curve point from a binary string.

\vspace{.5em}
\noindent\textsc{Step 3.}
It uses witness $w$ to compute $(y,z)$ in the third flow of $\mathit{PPET}$. It signs $(y,z)$
and $\mathcal{H}(k)$ with $sk_D$, and embeds both $(y,z)$ and the resulting signature $\sigma$ into a heartbeat
request message immediately after the handshake's end.

\vspace{.5em}
\noindent\textsc{Step 4.}
It passes all session keys to an external TLS library and is not involved in the rest of the
TLS connection in order to achieve better performance. $\Box$

\looseness=-1
Note that the triple $(y,z,\sigma)$ returned in Step 3 is the ODT in our scheme. Issuing a heartbeat request is
not a normal TLS client's behavior. However, it does not compromise obliviousness since O-TEE does this for all
TLS connections, except for TLS servers that explicitly disable heartbeat messages. We discuss alternatives
to heartbeat messages in Section~\ref{sec:TLS-compatibility}.

\looseness=-1
Since the hash of $\mathsf{HS}$ is covered by $\sigma$, a successful signature verification by the TLS
server implies that O-TEE, which owns $pk_D$, is also the TLS client holding $k$. This assures the TLS server
that O-TEE \emph{is} the communication endpoint with which it interacts and that the generated ODT is cryptographically
linked to this TLS connection. Revealing all TLS session keys \emph{after} signing $(y,z, \mathcal{H}(k))$ does
not break the ODT-TLS link. Although it can intercept the heartbeat request, \ADM cannot replace the ODT
because it cannot force O-TEE to generate a TLS shared key equal to $k$. Note that, to verify if \dn \emph{is} the device it
intends to check and also contains O-TEE owning $pk_D$, \T must validate $(y,z)$ using the verifier algorithm in $\mathit{PPET}$.

\subsubsection{\T\ as a TLS Server}
\label{sec:server}
\looseness=-1
Note that \N publishes all parameters used in witness generation and $\mathit{PPET}$,
such as group $\GG$ and generator $g$ so that \T has all the needed parameters
before starting any verification.

\looseness=-1
Suppose that \T sets up its own TLS server and waits for its agents' connection requests.
There are several ways for \T to obtain the expected value of the witness.
For instance, if \T trusts its command-and-control channel with the agent, it can use that
channel to secretly extract needed memory data. If the agent's memory state is less influenced
by external factors, \T can execute the agent on the system environment similar to the one in \dn.
The agent runs until it requests a TLS connection to \T's TLS server for ODT verification. The resulting
memory state mirrors the agent's state running on the remote devices. Note that address space layout
randomization \cite{pax2003pax,aslr-ccs04} used in commodity operating systems only randomizes
base addresses of code, stack, heap and libraries. Also note that these are only suggestions: our scheme
does not prescribe a specific way for \T to obtain the witness value.

Upon receiving a TLS connection request, \T proceeds as follows:

\noindent\textsc{Step 1.}
Following the TLS standard, it generates its Diffie-Hellman key share, computes $\mathsf{HS}$, and sets
\[
k=\mathsf{HKDF}(\mathsf{HS});c=\mathsf{Select}(k)\,.
\]
It obtains expected witness $w'$ and follows the $\mathit{PPET}$ protocol to compute the commitment
$v$ using the client nonce $n_0=u$ if it is in $\GG$. Otherwise, $v$ is computed with a randomly chosen $u$.
It then uses the Elligator decoding function $\iota$ to convert $v$ into
the 256-bit server nonce $n_1$ in its $\mathsf{ServerHello}$. After this step, it behaves exactly
according to the TLS standard.

\noindent\textsc{Step 2.}
If a heartbeat request message is received from the present TLS session,
\T responds to the heartbeat request following the TLS standard and saves the
ODT in the message for an offline verification in two steps described below.

(a) \T parses the binary string in the heartbeat request into $(y,z)$ and signature $\sigma$. It then
verifies $\sigma$ against $(y,z)$ and $\mathcal{H}(k)$ with the credential certified by \N.

(b) Following the $\mathit{PPET}$ protocol, \T checks if $z \stackrel{?}{=} y^s$ returns $\mathsf{True}$.
If so, it asserts that the device where its agent resides is \N-protected. $\Box$

\vspace{.3em} \looseness=-1
If $z\neq y^s$, \T cannot assert that the device is not \N-protected. There are various possible
reasons for that outcome. For instance, the ODT scheme's execution is attacked by an adversary
or the client process is not \T's agent.
Note that the agent can repeat the verification by making a new TLS connection.

\subsection{Experiments}
We built a prototype of the proposed ODT scheme to assess its performance and practicality.
All experiments were conducted on a laptop with Intel i5-10210U CPU and Manjaro Linux with the 5.15.155 kernel.

\subsubsection{Implementation}
To implement \T's TLS server, we integrate the operations specified in Section~\ref{sec:server}
with OpenSSL server version 3.2.0-alpha2. In addition to the OpenSSL's Curve25519 implementation, we use the
Elligator functionality of the Monocypher\footnote{https://monocypher.org/} library to implement
cryptographic parts of the protocol.

\looseness=-1
For the O-TEE, we use the Intel SGX SSL\footnote{https://github.com/intel/intel-sgx-ssl}
library that embeds OpenSSL in an SGX enclave.
We modify the TLS client functionality of this library for operations specified in Section~\ref{sec:client}.
An API is added to send the heartbeat request containing the ODT following the $\mathsf{Finished}$ flow.
Since the cryptographic setup of the ODT scheme is separate from the one for TLS, the scheme does
not restrict the cryptographic configuration of existing TLS clients and servers.

\subsubsection{Performance Results}
On the experimental platform, we set up two TLS clients (a native client and one using O-TEE) and
two TLS servers (a native server and an ODT server for verification). Since all TLS connections under
measurement are local, the results are dominated by computation time rather than network delays, which
more faithfully reflect the performance of the scheme. Both the ODT server and O-TEE sample the
heap in five locations and try to perform a verification for all incoming and outgoing connections respectively.

\looseness=-1
We run two experiments, each repeated 1,000 times, to get the averaged data. The first assesses how the
ODT server performs and how it differs from the native server. We use the $\mathsf{curl}$ tool to
connect to both servers and measure the time between sending $\mathsf{ClientHello}$ and receiving
$\mathsf{ServerHello}$. Note that although $\mathsf{curl}$ uses the native TLS client, the ODT
server still generates $v$ using a randomly chosen group element $u$. Table~\ref{tab:measurement-results}
reports the results. The ODT server incurs a slowdown of less than 1 millisecond per request due to parsing
the client nonce into $u$ and calculating the commitment $v$. Among the addition time spent by the ODT server,
the $\mathit{PPET}$ commitment generation and Elligator decoding cost 0.33 $ms$ and 0.17 $ms$, respectively.
The performance difference with the native TLS server does not compromise verification obliviousness.
A detailed analysis of side-channel attacks is presented in Section~\ref{sec:sidechannel}.

\looseness=-1
\begin{table}[ht]
    \centering
    \caption{Server response time measured using $\mathsf{curl}$ (in $ms$).}
    \begin{tabu}{rr}
    \tabucline[1pt]{-}
       ODT Server (overall cost) &  $8.25 \pm 0.71$ \\
        (a) PPET commitment & $0.33 \pm 0.26$ \\
        (b) Elligator decoding & $0.17 \pm 0.11$ \\ \hline
       OpenSSL Server  & $7.93 \pm 0.73$ \\
    \tabucline[1pt]{-}
    \end{tabu}
    \label{tab:measurement-results}
\end{table}

\looseness=-1
The second experiment measures the performance impact of O-TEE. We measure the time for a complete handshake
session in four combinations of client-server setups. The results are shown in Table~\ref{tab:handshake}.
Compared with the native TLS client, it costs O-TEE around $144ms$ more per handshake. The overhead is
due to cryptographic operations (including generating $u$, encoding $v$, and calculating the ODT),
witness generation, and extra time for SGX \texttt{ocall} and \texttt{ecall}.  While none of these costs
can be saved, we can reduce the handshake time by letting O-TEE send the $\mathsf{Finished}$ message
before ODT and witness generation. This optimization does not affect security of the scheme.

\begin{table}[ht]
    \centering
   \caption{Total handshake duration in four combinations of TLS connections (in $ms$)}
    \begin{tabu}{rlrr}
    \tabucline[1pt]{-}
       & & ODT server & OpenSSL Server \\ \tabucline[0.5pt]{-}
      O-TEE      &     & $164.30 \pm 2.81$  & $164.10 \pm 2.72$  \\
      OpenSSL Client & & $20.40 \pm 0.65$ & $20.08 \pm 0.63$  \\
      \tabucline[1pt]{-}
    \end{tabu}
    \label{tab:handshake}
\end{table}

\section{Security Analysis}
\label{sec:sec-analysis}
\subsection{Background on the Tamarin Prover}
Tamarin \cite{tamarin,cav13} is a tool for the mechanized analysis  of security protocols.
Tamarin works with a formal model of the security protocol, its desired properties, and the adversary.
It is used to construct either: (1) a proof that the protocol is secure, i.e., its properties hold
in the specified adversarial model (even when run with an unbounded number of protocol participants),
or (2) a counter-example, i.e., an attack on the protocol. Tamarin has been used for formal
verification of a wide range of large-scale real-world protocols \cite{tamarinIEEE2022}, including 5G
authenticated key agreement \cite{5G-2018},  TLS v1.3~\cite{tamarin-tls},
and the electronic payment standard EMV \cite{EMV-2021}.

\looseness=-1
Tamarin reasons using a symbolic model of cryptography, going back to the seminal work of Dolev and Yao
\cite{dolevSecurityPublicKey1983}. Protocols are represented as  labeled transition systems augmented with equational
theories formalizing common cryptographic operators.
This kind of model abstracts away from a low-level implementation of cryptography,
focusing instead on the properties of idealized cryptographic operators.

\subsection{Protocol Model in Tamarin}
To prove binding integrity, we first model the protocol.

\label{sec:prot-model-tamar}
\mypara{Protocol roles}
\looseness=-1
We specify five protocol roles: (1) an aggressor role that can create agents, (2) an agent role (\agent)
created by an aggressor and that can run on a device, (3) a device role that can run agents and be
assigned an O-TEE, (4) a neutral party (\N) role that can instantiate O-TEE bound to a device,
and (5) an O-TEE role.

\mypara{Adversary model}
The Tamarin adversary represents the dishonest \ADM. This adversary has all the capabilities of the
standard (Dolev-Yao) adversary used in symbolic models: it can read, create, modify, and block any
messages created or sent over the network. Moreover, it can compromise any device and gain its capabilities,
though it cannot compromise \N. On compromised devices, it can interrupt and resume the execution of an
attached O-TEE and read, create, modify, and block all messages the device sends or receives over
the network or from O-TEE. If it interrupts an O-TEE, the adversary can modify the data O-TEE reads.
As explained in Section~\ref{sec:mms}, the O-TEE measurement scheme ensures that the co-residence
witness generated by O-TEE after an interrupt is invalid. We also assume that the adversary can
mirror only a small subset of a running process's memory between devices instead of the whole memory.
We justify this using the argument presented in Section~\ref{sec:mms}.

\mypara{Measurements} \looseness=-1
As described in Section~\ref{sec:wg}, the core of our protocol is the witness generation scheme that has
two functions: $\mathsf{Select}$ and $\mathsf{Measure}$. While the protocol measures the agent's memory,
we model an abstraction of this and assume O-TEE measures \emph{properties} of a process.
Note that memory measurement is one way to achieve binding integrity and that other measurements,
such as cache or control flow measurements, are possible, provided they identify the process.

\looseness=-1
Given $\agent$ with a set of properties $\mathcal{P}(\agent)$, O-TEE measures a subset of these properties.
To select the subset, O-TEE uses the $\mathsf{Select}(k)$ function that takes as arguments a key $k$,
and returns a set $\mathcal{S}$ of positions that are measured. We do not model subsets explicitly and
instead leave them in their symbolic form $\mathcal{S} = \mathsf{Select}(k)$, where
$k$ represents the shared secret between O-TEE and the aggressor.

\looseness=-1
We model measurement as a function $\mathsf{Measure}(\mathcal{S}, \mathcal{P}(\agent))$ that returns a
witness $w$ given the set of indices $\mathcal{S}$ and the properties $\mathcal{P}(\agent)$ of an agent process.
As already mentioned, the adversary can interrupt O-TEE and recover the indices $\mathcal{S}$ or
modify the values O-TEE reads.

\mypara{Verification}
We model verification in the same way we specify it in our scheme. The aggressor first verifies that
signature $\sigma$ is valid and belongs to a registered O-TEE, and then verifies if the witness
response is correct. If both checks succeed, we consider the verification to be successful.

In summary, given public O-TEE measurement functions $\mathsf{Select}$ and $\mathsf{Measure}$, we define
$w = \mathsf{Measure}(\mathsf{Select}(k), \mathcal{P}(\agent))$ to be the witness extracted from an agent
process $\agent$ with properties $\mathcal{P}(\agent)$, given a key $k$ shared between the aggressor and O-TEE.
O-TEE's measurement and signature ensure that the aggressor can establish a binding between
its agent process, O-TEE, and the TLS connection.

\subsection{Binding Integrity}

\begin{figure}
    \centering

    \begin{lstlisting}[frame=single, basicstyle=\scriptsize\ttfamily]
lemma binding_integrity:
 "All AGR O-TEE AGT prop #i.
 AcceptedVerificationWithAgent(AGR, O-TEE, AGT, prop) @i
 ==>  (Ex NP D #j #k #l #m #n.
     OTEERegisteredOnDevice(NP, O-TEE, D) @j &
     OTEEHonestPropertyRead(O-TEE, D, AGT, prop) @k &
     AggressorCreateAgent(AGR, AGT, prop) @l &
     AggrDHKey(AGR, DHE) @m & OTEEDHKey(O-TEE, DHE) @n)"
    \end{lstlisting}

    \caption{Binding integrity property as specified in Tamarin.}
    \label{fig:binding-integrity}
\end{figure}

\looseness=-1
We formalize binding integrity from Section~\ref{sec:def} as a trace property in Tamarin as shown in Figure~\ref{fig:binding-integrity}.
We prove that a successful verification by an aggressor \textsf{AGR} that believes it is talking with an \textsf{O-TEE} and that created its own agent process \textsf{AGT} exhibiting properties \textsf{prop}, implies: 1) that the \textsf{O-TEE} is a genuine O-TEE registered on device \textsf{D} by a neutral party \textsf{NP}, 2) that the \textsf{O-TEE} measured the agent process \textsf{AGT} residing on device \textsf{D} that exhibited the properties \textsf{prop}, 3) the aggressor \textsf{AGR} created the agent \textsf{AGT} with properties \textsf{prop}, and 4) \textsf{AGR} and \textsf{O-TEE} derived the same Diffie-Hellman secret $\textsf{DHE}$.

\looseness=-1
In Section~\ref{sec:def} we note that the intrinsic challenge of ODT verification is to determine how entities to be authenticated are identified.
In our ODT scheme, the aggressor uses distinct injected agents to build a namespace of devices.
Since an agent's identity is ambiguous, i.e. there is no canonical way to refer to a piece of code, we use the properties that the code itself exhibits as the agent's identity.
As explained in Section~\ref{sec:mms}, our approach to witness generation guarantees co-residence:  when the measurement succeeds,
the process memory that was measured is on the same device as O-TEE.
This, combined with the fact the aggressor connects to an agent over a TLS channel, enables our ODT scheme to achieve binding integrity as proven in Tamarin.

\subsection{Verification Obliviousness}
\label{sec:verif-obliv}

\looseness=-1
We formalize verification obliviousness as the adversary's advantage over a random guess in differentiating a given process's two TLS handshakes: with a regular server and with \T performing ODT verification. Specifically, we define it as a game where the adversary (\ADM), given an \N-protected device $D$ and all secrets in its O-TEE, is challenged to determine if ODT verification took place with an oracle that simulates either a normal TLS server or \T depending on a random bit.
Note that verification obliviousness is not dependent on O-TEE's security.
The formal definition is below.

\begin{definition}[Verification Obliviousness]
\label{def:obl}
  Let $D$ be a device with O-TEE.
  Let \T be an aggressor and $p_{\mathcal{A}}$ be \T's agent on device $D$.
  Let $\OO(b)$ be an oracle that, with a random input bit $b$, simulates a standard TLS server for $b=0$ or \T otherwise.
  For all protocols $\mathit{ODT}$ and all adversaries $\adv$,
  we define the advantage function:
    \begin{equation*}
    \begin{split}
      \advantage{\mathsf{VO}}{\mathit{ODT}} = \bigg| & \frac{1}{2} - \PR \Big[ \T \riar p_{\mathcal{A}}; \{0, 1\} \riar b; \\
      & \adv^{\OO(b)}(D, p_{\mathcal{A}}) \riar b': b' == b \Big] \bigg|.
    \end{split}
  \end{equation*}

  An ODT protocol satisfies verification obliviousness if, for
  all efficient algorithms $\adv$ that have complete control over the device and its O-TEE, $\advantage{\mathsf{VO}}{\mathit{ODT}}$ is negligible.
\end{definition}

\looseness=-1
To show our scheme satisfies Definition~\ref{def:obl}, we first prove that the verifier's message $v$ of the $\mathit{PPET}$ protocol in Figure~\ref{fig:wv} is random in $\GG$.
This is a necessary condition for Elligator to encode elliptic curve points into uniformly random looking strings, a fact we use in our verification obliviousness proof.
We call this intermediate security property: \emph{verifier message uniformity}.

\begin{definition}[Verifier Message Uniformity]
  Let $\GG$ be a group of order $p$ and $g$ a generator of the group.
  Let $\WW$ be the domain of possible witnesses.
  For all PPET protocols $\mathit{PPET}$, all adversaries $\adv$,
  we define the advantage function:
  \begin{equation*}
    \begin{split}
      \advantage{\mathsf{VM\_UNI}}{\mathit{PPET}} = \; & \Big| \frac{1}{2} - \PR \Big[ \adv \riar u \in \GG; \ZZ_p \riar s; \\
         &  \WW \riar w'; v_0 = g^su^{w'}; \GG \riar v_1; \\
         & \{0, 1\} \riar b; \adv(v_b) \riar b': b' == b \Big] \Big|.
    \end{split}
  \end{equation*}
  A witness verification protocol $\mathit{PPET}$ satisfies verifier message uniformity if, for all
  efficient algorithms $\adv$, $\advantage{\mathsf{VM\_UNI}}{\mathit{PPET}}$ is negligible.
\end{definition}

\begin{lemma}
\label{lemma:verifier-message-uniformity}
  The Privacy-Preserving Equality Test Protocol $\mathit{PPET}$ in Figure~\ref{fig:wv} satisfies verifier message uniformity.
\end{lemma}
\begin{proof}
\looseness=-1
  Observe that $u^{w'}$ is an element of $\GG$.
  For a given uniformly random group element $g^s$, the probability that $g^su^{w'}$ is equal to a
  given group element is $1/p$. Since this holds for any possible value of $u^{w'}$,
  $v_0$ is uniformly distributed in $\GG$.
  Since $v_1$ is also uniformly distributed in $\GG$, the statistical difference between
  $v_0$ and $v_1$ is zero. Hence, $\advantage{\mathsf{VM\_UNI}}{\mathit{PPET}}$ is negligible.
\end{proof}

\begin{theorem}
\label{thm:obl}
  Our ODT scheme satisfies Definition~\ref{def:obl}.
\end{theorem}

\begin{proof}
In order to determine whether a verification took place, the adversary can run $p_{\mathcal{A}}$ and other processes of its choice on device $D$ and look at: (1) the oracle's response $v$, and (2) the behavior of O-TEE.
We discuss timing attacks and other side-channels in Section~\ref{sec:discussion}.

\looseness=-1
For (1), the only difference the adversary can detect between an honest TLS server and the aggressor is in how the $\textsf{ServerHello}$ nonce is generated.
When the aggressor sends $\iota^{-1}(v)$, the nonce's distribution might be different from a uniformly random string that is sent by a normal TLS server.
However, according to Lemma~\ref{lemma:verifier-message-uniformity}, we know that $v$ is uniformly random in $\GG$.
After $v$ is Elligator encoded, we are guaranteed that an adversary has a negligible probability to differentiate the Elligator-encoded $v$ from a uniformly random looking string.
For (2), the adversary cannot use O-TEE as an oracle to test if a given process is performing verification, as O-TEE always behaves in the same way for all processes that invoke a TLS connection.
Therefore, we conclude that $\advantage{\mathsf{VO}}{\mathit{ODT}}$ is negligible.
\end{proof}

\looseness=-1
The proof shows that \ADM cannot use our ODT protocol to detect a verification attempt.
It does not capture the probability that \ADM can detect the agent.
The likelihood of agent detection ultimately depends on the arms race between \T and \ADM, whose outcome is not prescribed by our design.
We show below that detecting an agent for ODT verification is as hard as detecting TLS-using malware.

\looseness=-1
\mypara{Discussion on Agent Detection}
According to several surveys \cite{OR19,oh2021survey}, it is popular for malware to use TLS to communicate with its backend server due to the widely used HTTPS protocol. For instance, the TLS-based Cobalt Strike Command \& Control channel was widely used by ransomware as reported by CISCO \cite{CISCO20} and was also later found in the SolarWinds attack \cite{Fireye}. Accompanying this trend is the evolution of malware detection techniques based on TLS anomalies, benefiting from advances in machine learning. Traffic data often used in anomaly analysis includes port numbers, the amount of transported data, connection duration, and handshake parameters \cite{AM16}. The potential anomaly exhibited by malware is due to the involved task, such as downloading a payload or uploading stolen data. Since the ODT verification agent does not have such tasks, it can avoid detection by making ordinary TLS connections.

\looseness=-1
Our ODT scheme imposes no requirements on an agent's memory contents. Thus, memory contents involved in O-TEE's measurement do not have a specific pattern for \ADM to fingerprint an agent.
\T can utilize existing anti-detection techniques  \cite{ANSB19}, undocumented evasion techniques, and zero-day exploits to conceal its agent's existence. The mere existence of a TLS connection to a rarely-accessed server could appear suspicious in certain scenarios. \T may evade such detection by covertly compromising existing processes that use TLS and connecting them to its servers providing legitimate Internet services (e.g., DNS or NTS).

\looseness=-1
Lastly, we acknowledge that not all aggressors can invade target devices with a stealthy agent performing TLS connections. However, it is feasible for a state-backed aggressor, which is the primary subject of the ICRC's emblem application and international humanitarian law, to have such a capability.
We leave it to future work to investigate ODT schema with less demanding requirements.

\subsection{Security Preservation}
\label{sec:sec-preservation}
Security preservation concerns a TLS server that extracts information about $D$ from its TLS connections. We formulate it as the adversary's winning probability in the game where it guesses the value stored at a specific memory location.
We quantify the probability that the adversary guesses correctly in relation to the memory's size, the number of memory locations O-TEE measures, the number of values each position can have, and the number of adversary queries.
In the proof, we show that this probability is negligible.

\begin{definition}[Security preservation]
\label{def:sec-preservation}
  Let $D$ be a device with O-TEE. Let $X$ be the set of values that can be stored at a memory location.
   Let $p$ be a process on $D$ with memory contents $(x_i)_{i \in I}$, where $I$ is the index set of memory locations and $x_i \in X$ is the value stored at location $i$.
  Let $\OO_p$ be an oracle that, for the process $p$, simulates the execution of O-TEE on device $D$ that measures $C$ locations up to $q$ times.
  For all protocols $\mathsf{ODT}$, all adversaries $\adv$ that output a set $\{ x'_{j_1}, x'_{j_2}, \dots, x'_{j_C} \}$ where $x'_j \in X$, we define the advantage function:
  \begin{equation*}
    \begin{split}
      \advantage{\mathsf{SP}}{\mathit{ODT}} = & \;
      \PR[ I_t = \{i_1, i_2, \dots, i_C\}, \forall i \in I; \\
      & \adv^{\OO_p}(I_t) \riar \{ x'_{i_1}, \dots, x'_{i_C} \} : \forall i \in I_t. x'_i == x_i ].
    \end{split}
  \end{equation*}
\looseness=-1
An ODT protocol satisfies security preservation if, for
  all efficient algorithms $\adv$, $\advantage{\mathsf{SP}}{\mathit{ODT}}$ is negligible.
\end{definition}

\begin{theorem}
\label{thm:sec-pre}
  Our ODT scheme satisfies Definition~\ref{def:sec-preservation}. $\Box$
\end{theorem}

\looseness=-1
Our proof (in the appendix) considers two different attack scenarios: extracting information from a low-entropy process, and extracting information about a 256-bit secret key with the knowledge of contents at all other memory locations.
For the first scenario, we show that $\adv$ only has a negligible success probability given that $C$ is sufficiently large.
We can lower $\adv$'s success chance by measuring more locations.
However, this demands \T to have a more complete and accurate view of the agent's memory.

\looseness=-1
For the second scenario, assuming that O-TEE measures one 64-bit word, we show that $\adv$ learns at least one of the four words of the 256-bit secret key with a small ($2^{-59}$) probability after making $10^6$ guesses for a process with one megabyte of memory.
The relatively high probability is not an issue considering $\adv$ relies on the device to initiate a TLS connection with its server.
This scenario also represents an ideal case for $\adv$ where it knows all other memory locations, which is generally infeasible in practice.
Using more locations to measure also lowers the probability.

\section{Discussions and Related Work}
\label{sec:discussion}
\subsection{Alternatives to TLS v1.3}
\looseness=-1
Among the protocols in the TLS family, the ODT scheme is compatible with DTLS v1.3 \cite{DTLS13},
a variant of TLS v1.3 for secure UDP communications.
However, TLS v1.2 \cite{TLS12} (and its DTLS variant) do not fit the ODT scheme well. Servers using TLS v1.2
typically use public-key encryption-based handshakes where $\mathsf{ServerHello}$ is sent before the handshake
secret is generated. As a result, \T cannot send its commitment $v$ as a server nonce and instead
has to send it in TLS records that carry application data. This undermines verification obliviousness
since there is no uniform pattern for TLS servers' application data sending and O-TEE has to respond differently.
It is thus seems infeasible to prove that \T's message is indistinguishable from those sent by regular TLS servers.

\looseness=-1
Although there are other protocols involving random data exchanges and shared keys, none of them is more suitable
for the ODT scheme than TLS. Application layer protocols (such as Tor \cite{TOR}, IRC, and SSH) are not generic
enough for all applications. Network layer protocols (such as the Internet Key Exchange protocol \cite{IKE} for IPSec)
are not application specific; hence, it is difficult to link the connection to a particular process.

\subsection{Compatibility with TLS Implementation}
\label{sec:TLS-compatibility}
\T's server is fully compatible with the TLS specification and provides services to all TLS clients.
However, we discovered that the heartbeat request sent by O-TEE in our prototype disrupts TLS
connections with servers that use OpenSSL.
Although TLS v1.3 \cite{TLS13} defines the heartbeat extension, OpenSSL maintainers
removed it from the code due to lack of real-world use-cases and concerns over implementation
bugs.\footnote{\url{https://github.com/openssl/openssl/issues/4856}}
We believe that the ICRC emblems are an important real-world use-case that would warrant the
re-introduction of heartbeats into OpenSSL.

\looseness=-1
To avoid the compatibility issue, an O-TEE implementation can instead use an ICMP Echo-Request messages
to deliver ODTs to TLS servers. This alternative does not compromise binding integrity, since
our scheme does not depend on the security of the TLS connection to protect this message.

\subsection{Side-Channel Attacks on Verification Obliviousness}
\label{sec:sidechannel}
\looseness=-1
We briefly consider the risks of compromising verification obliviousness via side-channel attacks
on \T's messages and O-TEE. Since \T is remote from \ADM, the only side-channel observable to \ADM is
based on time intervals between sending $\mathsf{ClientHello}$ and arrival of $\mathsf{ServerHello}$.
However, \T can introduce an artificial delay in its non-ODT computations to mask the difference.
Unless \ADM knows the hardware that \T's server runs, it cannot detect this delay.

\looseness=-1
On the client side, the proposed ODT scheme eliminates the device's side-channel leakage about
ODT verification. O-TEEs behave in the same fashion regardless whether the TLS peer is a regular server or \T.
Furthermore, the scheme does not require \T's agent to carry out special operations for ODT verification.
For example, it does not perform a local attestation against an enclave O-TEE. While we acknowledge
that there are side-channel attacks against secrets of various TEEs, coping with them
is orthogonal to this work.

\looseness=-1
\subsection{Limitations of Our ODT Scheme}
\label{sec:limitations}
Our scheme cannot be deployed on devices without a hardware-based TEE or virtual machines in a cloud even
if confidential VM techniques such as Intel TDX are in place.
It also cannot help aggressors without an agent on the target devices.
Note that such aggressors can resort to ADEM \cite{adem} to perform preliminary checks.
Also, \N can deploy ADEM and ODT side by side.

\looseness=-1
The ODT schemes guarantees that if an ODT passes \T's verification, the device is \N-protected.
However, the converse is not guaranteed. Thus, \T cannot conclude that the device is not \N-protected
if verification fails. While it seems impossible to prove an ODT's non-existence, it is an open problem
how to use the ODT verification transcript as an unforgeable proof for \T's due diligence.

\looseness=-1
Another limitation is the dependency on TLS v1.3, especially for the Diffie-Hellman handshake.
It may require changes to work with future versions of TLS or a successor protocol suite. However,
if the handshake algorithm is unchanged, we are optimistic that the adaptation will only involve message
encoding, rather than changes of core algorithms.

\subsection{Related Work}
As mentioned earlier, ADEM~\cite{adem} proposed an emblem based on digital signatures for use in network
communications. While it ensures verification obliviousness, the resulting emblem-device binding is
insecure against a malicious kernel.

\looseness=-1
The functionality of ODT is related to RoT attestation whereby a prover device attests its RoT attributes
to a trusted verifier. These attributes include RoT genuineness \cite{tpm,dcap}, RoT presence \cite{presence_att},
and the distance \cite{CRSC12,dhar2020proximitee} to the verifier.
More specifically, our protocol is related to remotely attested TLS protocols \cite{TC4SE,RATLS,IntelRA-TLS,TSL}.
These protocols combine TLS with remote attestation to create \emph{trusted channels}, i.e. secure channels
where one or both endpoint(s) is/are attested to the other. However, they are incompatible with
verification obliviousness.

\looseness=-1
In terms of verification obliviousness, the ODT scheme is related to privacy protection techniques for hiding
sensitive information that would otherwise be exposed to observers monitoring system or network actions.
Related work includes oblivious RAM (ORAM) \cite{ORAM} and Private Information Retrieval (PIR) \cite{pir},
which prevent an entity's interest in particular data items from leaking due to its actions on the data set.
Anonymous communications techniques such as Tor \cite{TOR} and mix networks \cite{jakobsson2001optimally},
prevent sender or receiver information leakage from network activities.

\looseness=-1
The proposed PPET protocol provides another type of privacy protection, which primarily considers the privacy of
data used in multi-party computation. Similar techniques include oblivious transfer \cite{OT}, zero-knowledge proofs \cite{goldreich1994definitions}, secret handshakes \cite{BDSSS03} and private set intersection \cite{PSI,PSZ18,PSI99,PSI_tsudik}.

\looseness=-1
\section{Conclusions}
\label{sec:conc}
The need for digital emblems arises in the context of cyber-conflicts and international
humanitarian law. Developing effective solutions triggers novel requirements for labeling devices
in a way that provides binding integrity, while ensuring obliviousness for the verifier and
security preservation for the provers. In this paper, we defined this problem, fleshed out
its requirements, and provided the first solution, along with a prototype implementation.
Future work includes experiments with the proposed scheme to further assess its practicality and
exploring open problems, such as  ascertaining the absence of oblivious digital tokens.

\section*{Acknowledgments}
Mihael Liskij's and David Basin's research was funded by the Werner Siemens-Stiftung (WSS) as part of the Centre for Cyber Trust (CECYT). We thank the WSS for their contribution.

Gene Tsudik's work was supported in part by NSF Award SATC1956393 and NSA Award H98230-22-1-0308.

This research was also partially supported by the Lee Kong Chian Fellowship awarded to Xuhua Ding by Singapore Management University.

\section*{Ethics Considerations}
This research is directly motivated by the ICRC's initiative to introduce digital emblems in cyberspace.
Such emblems should provide a way for protected parties to signal their right to protection in cyberspace,
analogous to the use of a red cross, crescent, or crystal to mark, e.g., medical relief workers and medical
facilities/objects in the physical world. This, in turn, provides a basis for extending international
humanitarian law to cyberspace and making the world safer, in wartime.
Our research focused on the technical foundations for digital emblems: their requirements, design,
implementation, evaluation, and security proofs.
There are no ethical concerns, since this work neither involved human subjects nor attacked
systems or infrastructure of any kind. Hence, there was no need for approval by the authors' institutions'
ethical review boards.

\section*{Open Science}
The Tamarin model, Tamarin proofs, the ODT prototype, and all raw measurement results are available for
download at Zenodo \cite{ODT-zenodo}.

\printbibliography

\begin{appendices}

\section*{Appendix A: Theorem~\ref{thm:privacy-preservation} Proof}
\label{app:ppet-proof}
\begin{proof}
  We use a hybrid proof to show that the PPET protocol $\mathit{PPET}$ from Figure~\ref{fig:wv}
  satisfies privacy preservation. We first define starting and end distributions:

  \begin{equation}
    \begin{split}
      \DD_b = \{& \GG \riar u; \ZZ_p \riar t; \WW \riar w; \adv(u) \riar v \in \GG; \\
                & z_0 = v^t u^{-wt}; \GG \riar z_1 : (g^t, z_b) \}.
    \end{split}
  \end{equation}

  Note that $\DD_0$ exactly matches the inputs given to the
  adversary in the privacy preservation game when $b = 0$ and
  $\DD_1$ matches the inputs when $b = 1$. From this, we have
  $\advantage{\mathsf{PP}}{\mathit{PPET}} =
  \advantage{\Delta}{\DD_0, \DD_1}$. To prove that the
  protocol is secure, we show that $\DD_0$ and $\DD_1$ are
  computationally indistinguishable.

  \mypara{First hybrid} We define $\HH_1$ and show that it is
  statistically indistinguishable from $\DD_0$:
  \begin{equation}
    \begin{split}
      \HH_1 = \{& \GG \riar u; \ZZ_p \riar t; \WW \riar w; \adv(u) \riar u^{x}; \\
                & z = u^{xt - wt} : g^t, z \}.
    \end{split}
  \end{equation}

  By the protocol, $v$ must be an element from $\GG$. Therefore,
  there exists an $x$ such that $u^{x}$ is equal to $v$. Hence this is just a syntactic change and we have
  $\advantage{\Delta}{\DD_0, \HH_1} = 0$.

  \mypara{Second hybrid} We define $\HH_2$ and show that it is
  statistically indistinguishable from $\HH_1$:
  \begin{equation}
    \begin{split}
      \HH_2 = \{& \GG \riar u; \ZZ_p \riar t; \WW \setminus x \riar w; \adv(u) \riar u^{x}; \\
                & z = u^{xt - wt} : g^t, z \}.
    \end{split}
  \end{equation}

  The only change is that we are sampling $w$ from a smaller set. The
  adversary can distinguish $\HH_1$ and $\HH_2$ with an advantage of
  at most the probability of outputting such an $x$. In the real
  world, this advantage represents the adversary's a priori knowledge
  of the distribution of the witness values. However, as stated in
  Definition~\ref{def:witn-indist}, we perform this proof under the
  assumption that the adversary guesses the witness
  incorrectly. Therefore, we consider the advantage of the adversary
  in this step to be zero, i.e.
  $\advantage{\Delta}{\HH_1, \HH_2} = 0$, when the adversary guesses the witness incorrectly.

  \mypara{Third hybrid} We define $\HH_3$ and show that it is
  statistically indistinguishable from $\HH_2$:
  \begin{equation}
    \begin{split}
      \HH_3 = \{& \GG \riar u; \ZZ_p \riar t; \adv(u) \riar u^{x}; \\
                & \WW \setminus 0 \riar k; z = u^{kt} : g^t, z \}.
    \end{split}
  \end{equation}

  This is again just a syntactic change where we rewrite
  $k = x - w \in \WW \setminus 0$, therefore
  $\advantage{\Delta}{\HH_2, \HH_3} = 0$. Note that the distribution
  of $k$ depends on the adversary's distribution of witnesses embedded
  in $x$. Because the adversary might have some a priori knowledge on
  the distribution of witnesses, we cannot assume $k$ is sampled uniformly.

  \mypara{Fourth hybrid} We define $\HH_4$ and show that it is
  computationally indistinguishable from $\HH_3$:
  \begin{equation}
    \begin{split}
      \HH_4 = \{& \GG \riar u; \ZZ_p \riar t; \adv(u) \riar u^{x}; \\
                & \WW \setminus 0 \riar k; \ZZ_p \riar c; z = g^{kc} : g^t, z \}.
    \end{split}
  \end{equation}

  We uniformly sample $c$ and use it instead of $t$ to compute $z$. We
  do a reduction to the Decisional Diffie-Hellman (DDH) problem and
  show that $\advantage{\Delta}{\HH_3, \HH_4}$ is negligible. We
  assume that there exists an algorithm $\adv$ that can efficiently
  distinguish $\HH_3$ and $\HH_4$. We build an algorithm $\RR$ that
  uses $\adv$ to break $\mathsf{DDH}$.

  $\RR$ receives from its DDH challenger
  $\alpha = g^a, \beta = g^b, \gamma = g^c$.
  It simulates $\adv$'s game and uses $\alpha$ instead
  of $u$, $\beta$ instead of $g^t$, and $\gamma$ instead of
  $u^{t}$.
  When $b = 0$ in $\RR$'s game, $\RR$ outputs $(\beta, g^{kab})$ that exactly matches the distribution $\HH_3$.
  When $b = 1$ in $\RR$'s game, $\RR$ outputs $(\beta, g^{kc})$ that exactly matches the distribution $\HH_4$.
  If $\adv$ outputs $b' = 0$ as its own guess $\RR$ outputs 0, and if $\adv$ outputs $b' = 1$ it outputs 1.
  Because $\RR$ perfectly simulates $\adv$'s distinguishing game, we
  conclude that
  $\alg{adv}^{\mathsf{DDH}}_{\RR} = \advantage{\Delta}{\HH_3, \HH_4}$.
  However, this is a contradiction because DDH is difficult in the
  group $\GG$. Hence $\advantage{\Delta}{\HH_3, \HH_4}$ is
  negligible.

  \mypara{Last step} We show
  $\advantage{\Delta}{\HH_4, \DD_1} = 0$. For convenience,
  $\DD_1$ is:
  \begin{equation}
    \begin{split}
      \DD_1 = \{ \GG \riar u; \ZZ_p \riar t; \adv(u) \riar v;
                 \GG \riar z_1 : (g^t, z_1) \}.
    \end{split}
  \end{equation}

  While the distribution of $k$ is not uniform, combining it with a
  uniformly random variable $c$ ensures that $g^{kc}$ is uniform. To
  prove that we use the fact that the order of $\GG$ is prime. For any
  value of $k$, $g^k$ is a generator of the group. Because $c$ takes
  all possible values from $\ZZ_p$, the probability that $g^{kc}$ is
  equal to any group element is $1 / p$. Since this holds for
  every $k$ (no matter the distribution of $k$), we can conclude that
  the probability that $g^{kc}$ is equal to a specific group element
  is also $1 / p$. Therefore, the distributions $\HH_4$ and $\DD_1$
  are statistically indistinguishable.

  \mypara{Conclusion} We have shown for all distributions in our chain
  that they are either computationally or statistically
  indistinguishable from their neighbors. Using the hybrid argument,
  we conclude that $\DD_0$ and $\DD_1$ are computationally
  indistinguishable and therefore the advantage
  $\mathsf{adv}^{\mathsf{PP}}_{\adv}$ of the adversary $\adv$ in
  winning the $\mathsf{PP}$ game is negligible.
\end{proof}

\section*{Appendix B: Theorem~\ref{thm:sec-pre} Proof}
\label{app:thm3proof}
We argue this,
considering two scenarios:
(i) the adversary has partial knowledge of the memory and (ii) the adversary has full knowledge of the entire memory except for a 256 bit secret.
{The first scenario models a general brute force attack where $\adv$ attempts to extract information from a system that has a relatively low entropy.
We show that $\adv$ has a negligible advantage given that the number of measured locations $C$ is our security parameter.
The second scenario models a targeted attack where the adversary is interested in extracting a cryptographic key from a process's memory.
We use these cases to show that our scheme remains practically secure even against extremely powerful adversaries.
}

(i) Assume, for simplicity, that O-TEE measures 64 bit locations and that the adversary knows $k$ out of 64 bits for every measurable location but not all of them.
For the first measurement, the adversary has a probability of $(|X|-2^k)^{-C}$ to guess correctly.
After $l$ guesses over the same subset of measured values, the probability of the next guess being correct increases to $\big((|X|-2^k)^{C} - l\big)^{-1}$.
On average, each subset is measured $q\binom{|I|}{C}^{-1}$ times.
To simplify the calculations, we assume the adversary correctly guesses with the probability of the last guess $\Big(\big(|X|-2^k\big)^{C} - q\binom{|I|}{C}^{-1}\Big)^{-1}$, which provides an upper bound on the adversary's chance of success.
The probability $P$ that the adversary wins in the game by correctly guessing at least once using $q$ queries is then
\begin{equation*}
  \begin{split}
    P = 1 - \Bigg(1 - \bigg(\Big(|X|-2^k\Big)^{C} - q\binom{|I|}{C}^{-1}\bigg)^{-1}\Bigg)^q \,.
  \end{split}
\end{equation*}
Note that we interpret the adversary's output as an additional query and absorb it into $q$. This shifts the domain of $q$ by one and ensures that $q$ is greater than zero.

A function is defined to be negligible if $\lim_{n \rightarrow \infty}f(n)n^s = 0$ for all $s > 0$.
We set $f(n) := P(C)$ and show that $\lim_{C \rightarrow \infty}P(C)C^s$ goes to zero when $C$ is our security parameter.
First, we simplify $P(C)$ to $\bar{P}(C) = 1 - (1 - (2^C-q)^{-1})^q$ by setting $\binom{|I|}{C}^{-1}$ to one and $|X| - 2^k$ to two.
Both of these changes do not diminish the adversary's winning probability, hence $P(C) \leq \bar{P}(C)$.
Let $e = (2^C-q)^{-1}$, we calculate:
\begin{equation*}
  \begin{split}
    \bar{P}(C) = (1 - (1 - e)^q) = e \sum_{i=0}^{q-1}(1 - e)^i < e q = q (2^C-q)^{-1}.
  \end{split}
\end{equation*}
It is easy to see that $\lim_{C \rightarrow \infty} q (2^C-q)^{-1} C^s = 0$.
Since $P(C)$ is bounded from above by $q (2^C-q)^{-1}$, we have $\lim_{C \rightarrow \infty}P(C)C^s = 0$ and therefore the probability the adversary wins in the security preservation game is negligible.

For the second scenario (ii), we assume that the 256 bit secret is aligned to the measured positions, i.e. four positions cover it fully, that O-TEE measures only one 64 bit position, and that the adversary knows all memory contents except the secret.
We slightly modify the security game by fixing $I_t$ to be the set of four locations containing the key.
The adversary wins if it guesses correctly at least one of the four locations.

\looseness=-1
The probability of correctly guessing the value of the key after $l$ guesses on the same position is $(|X| - l)^{-1}$.
Each position will be measured $q / |I|$ times on average.
Similarly to the previous scenario, we assume that the adversary is always guessing with the probability of the last guess giving us the formula $(|X| - q / |I|)^{-1}$.
The probability of the adversary winning the game by guessing correctly at least one of the four key positions using $q$ queries is then
\begin{equation*}
    \begin{split}
        1 - \Big(1 - \big(|X| - q / |I|\big)^{-1}\Big)^{4q/|I|} \,.
    \end{split}
\end{equation*}

For a one megabyte program, the adversary has a probability of $2^{-59}$ of winning in the game {assuming it can make} $10^6$ requests.
While this probability is relatively high for cryptographic standards, we believe this is not an issue in practice.
With the results from our first scenario, we can increase the security of our scheme by measuring more locations.
Moreover, the adversary does not have full freedom of establishing requests and must rely on the victim process to initiate a connection.
Finally, an adversary that has perfect knowledge of every memory location at every moment a TLS connection opens is highly unlikely.
\end{appendices}

\end{document}